\documentclass[aps,twocolumn,pra,superscriptaddress,nofootinbib]{revtex4-2} %
\pdfoutput=1

\usepackage{graphicx}
\usepackage{float}
\usepackage{amsthm, amsmath, amssymb, amsfonts,bm,physics, dsfont}
\usepackage{changes}
\usepackage{mathtools}
\usepackage{hyperref}
\usepackage{natbib}

\newtheorem*{proposition*}{Proposition}

\def\C{{\mathbb C}}
\def\R{{\mathbb R}}
\def\id{{\mathbb I}}

\def\reqbit{rebit}

\def\opone{\leavevmode\hbox{\small1\normalsize\kern-.33em1}}
\def\tr{\mbox{tr}}
\newcommand{\cE}{\mathcal{E}}

\newcommand{\ii}{\mathrm{i}}

\newcommand{\cB}{\mathcal{B}}
\newcommand{\cS}{\mathcal{S}}

\begin{document}

\title{{Partial independence suffices to rule out   Real Quantum Theory experimentally}}

\author{Mirjam Weilenmann}
\affiliation{Inria, Télécom Paris - LTCI, Institut Polytechnique de Paris, 91120 Palaiseau, France}
\affiliation{Department of Applied Physics, University of Geneva, Switzerland}
\author{Nicolas Gisin}
\affiliation{Department of Applied Physics, University of Geneva, Switzerland}
\affiliation{Constructor University, Bremen, Germany}
\author{Pavel Sekatski}
\affiliation{Department of Applied Physics, University of Geneva, Switzerland}
 
\begin{abstract}
The role of complex quantities in quantum theory has been puzzling physicists since the beginnings. {It is thus natural to ask whether, in order to describe our experiments, the  mathematical structure of complex Hilbert spaces it is built on is really necessary.} Recently, it was shown that {this structure} is inevitable in network scenarios with independent sources. {More precisely}, Real Quantum Theory cannot explain the predictions of (Complex) Quantum Theory [Renou et al., Nature 600, 2021]. Here, we revisit the independence assumption underlying this work. We show that assuming partial independence is sufficient for showing the inadequacy of Real Quantum Theory. We derive a tradeoff between source independence and the Bell value achievable in {Real Quantum Theory}, which also lower bounds the source correlations required to explain previous experiments {by means of} real quantum systems. We further show that $1$ bit of entanglement is necessary and sufficient for recovering the complex quantum correlations by means of Real Quantum Theory in the scenario from [Renou et al., Nature 600, 2021].  Finally, building on [McKague et al., PRL 102, 2009], we provide a construction to simulate any complex quantum setup with $m$ independent sources by means of {Real Quantum Theory}, by allowing the sources to share a $m$ real-qubit entangled state in the first round of the experiment.
\end{abstract}

\maketitle

\section{Introduction}

Numbers are there to count and essential to reason about the world. Initially, integers could count sheeps, flowers and hectares. Later came astronomy and the desire to do calculus with continuous quantities -- time, positions and velocities of objects. This naturally led to extend numbers to what is known today as the real numbers. Regardless of their role in describing physical reality, people started playing with numbers for the sake of their intrinsic beauty. For instance one can easily see that some real numbers are irrational (not a fraction of integers), e.g. $\sqrt 2$ the solution to the equation $x^2=2$. In turn, other equations like $x^2=-1$ do not seem to have a solution. 
But equations without solutions are not elegant. Hence, mathematicians invented new numbers, today called complex numbers, which give a solution to all polynomial equations. A complex number $z=x+\ii y$ is the sum of a real and an imaginary part, where $x$ with $y$ are real and the new number $\ii$ is precisely here to solve the equation $\ii^2=-1$. The name imaginary was coined by Descartes to emphasize his doubts towards the existence of quantities corresponding to such numbers \cite{blank1999imaginary,merino2006short}. He was not alone to express such feelings. Later, Euler referred to imaginary numbers as ``impossible''\cite{blank1999imaginary}. Even Cauchy, the father of modern complex analysis, expressed his repudiation for the imaginary unit $\ii$\cite{blank1999imaginary}.

Physics has for a long time been based on purely real quantities. Textbooks on classical mechanics and thermodynamics usually don`t even mention complex numbers, since these theories do not need them, even Maxwell’s equations make no use of complex numbers.  Certainly, complex numbers are commonly used when describing electromagnetic waves. However, they are typically introduced as a powerful mathematical tool to facilitate computations, but without the ambition to correspond to physical quantities: electric and magnetic field are characterized by real valued quantities. The same can be said in relativity theory, where complex numbers are sometimes useful as a tool to express Lorentz transformations as generalized rotations in 4-dimensional space-time. 

The situation is seemingly different in quantum theory, where the quantum state itself is complex valued. 
In fact, the first postulate of quantum theory already associates to a physical system a complex Hilbert space $\C^d$ of dimension $d$. But is this 
{mathematical structure} needed to describe our experiments? Or would real Hilbert spaces, 
leading to a real-valued description of quantum states, do?

In fact, to describe a quantum experiment with a real Hilbert space, it suffices to add an auxiliary 2-dimensional subsystem, a real-qubit (rebit)~\cite{stueckelberg}. This 
also works when applied to composite systems. However, one encounters a problem when attempting a physical interpretation of this structure -- it is not clear to which subsystem the auxiliary \reqbit\, should be attached. 
This is particularly relevant when the subsystems are at a distance from each other, but should all be able to access this auxiliary \reqbit. {Note that the issue comes up because in Real Quantum Theory, defined in this work like in the previous literature~\cite{real_quantum1,real_quantum2, mckague2009simulating, Chiribella2010, Hardy2011, Barnum2020}, -- just like in Complex Quantum Theory -- composition of independent systems is captured by taking the 
{standard} tensor product of the individual independent states and measurements. This has to be satisfied by all involved systems, including states of and measurements on the auxiliary system.}
{Note that Real and Complex Quantum Theory are thus different theories. This is contrasted by the works~\cite{hita2025quantum, hoffreumon2025quantum}, postdating ours, which both change the way combined systems are described in order to reformulate Complex Quantum Theory in terms of real
Hilbert spaces (see Note added for details).}

The solution {to this problem} has been found over a decade ago \cite{mckague2009simulating}: this \reqbit \ can be delocalized among auxiliary qubits attached to each subsystem, in the form of a carefully encoded logical system (see the details below). 
For years this was considered to settle our questions --  the structure of {Complex Quantum Theory} is not needed (even though convenient mathematically), instead {Real Quantum Theory describes our experimental findings accurately too.}

However, recently it was realized that the above construction is incompatible with the assumptions made in network scenarios. More precisely, in any experiment where several \emph{independent} sources prepare different subsystems, admitting the presence of an entangled state shared among all subsystems is unsatisfactory. To highlight this insight, \cite{renou2021quantum} proposed a simple experiment with two sources and three observers, leading to correlations that admit no description by {Real Quantum Theory}  consistent with their assumption of source independence). The assumption that the sources do not share any entanglement is central to this conclusion, without it, {Real Quantum Theory} could not be ruled out. But what if the sources are only partially independent? This is what we investigate in this article.

We quantitatively bound the amount of quantum correlations that is needed {according to Real Quantum Theory} 
to reproduce the experimental results predicted by quantum theory in the setup proposed in~\cite{renou2021quantum}. We show that if the implementation achieves the maximal Bell score (in Complex Quantum Theory), the quantum correlations necessary in Real Quantum Theory are also sufficient to reproduce any experiment in the same causal network. Finally, considering the non-ideal setting we relate this minimal amount of entanglement needed in the real quantum model to the Bell score observed in the experiment. 
\\

\section{{Complex and Real Quantum Theories}}
In non-relativistic quantum theory, the totality of degrees of freedom needed to describe a physical system $C_i$ is associated with a complex Hilbert space $\C^{d_i}$, that we take to be finite dimensional for the sake of simplicity. Hence in a setup with $n$ systems, their joint state is described by a normalized vector on the product of all these spaces
$   \ket{\Psi}_{\bm C} \in \C^{d_1}\otimes \dots \otimes\C^{d_n}, $
or a density operator $\rho_{\bm C}$ acting on this Hilbert space, where $\bm C = C_1 \dots C_n$ denotes the collection of all the systems. {In particular, a global state of systems prepared independently is given by $\ket{\Psi}_{\bm C} = \ket{\psi_1}_{C_1}\otimes\dots \otimes \ket{\psi_n}_{C_n}$.} Measurements and transformations are also described by the POVMs and maps, with respect to this global Hilbert space. Their action must respect the locality of the systems -- a joint operator can be essentially instantaneously performed on two systems only if they are localized in the same region of space. This is what we call Complex Quantum Theory.

Real Quantum Theory is defined here analogously to the above, with the important difference that the complex Hilbert spaces are replaced by real ones. {This definition is in line with the previous literature~\cite{real_quantum1,real_quantum2,mckague2009simulating, Chiribella2010, Hardy2011, Barnum2020}, but differs from that of the works~\cite{hita2025quantum,hoffreumon2025quantum} which postdate ours.} For clarity let us here denote the systems with $\bar R_i$, their global state is then an element of
$\ket{\Psi}_{\bar{\bm  R} } \in \R^{d_1' }\otimes \dots \otimes\R^{d_n'},$
{and independent systems are described by a product state.}
The dimensions of the real Hilbert spaces must of course be bigger to accommodate the description of the same degrees of freedom. 

In the following, we recall the thought experiment from~\cite{renou2021quantum} that shows that in a scenario with two independent sources no strategy in Real Quantum Theory can explain the predictions of Complex Quantum Theory.

\section{Thought experiment of \cite{renou2021quantum}}

We consider the setup of Figure~\ref{fig:bilocal_setup}.
\begin{figure}
   \centering   \includegraphics[width=\columnwidth]{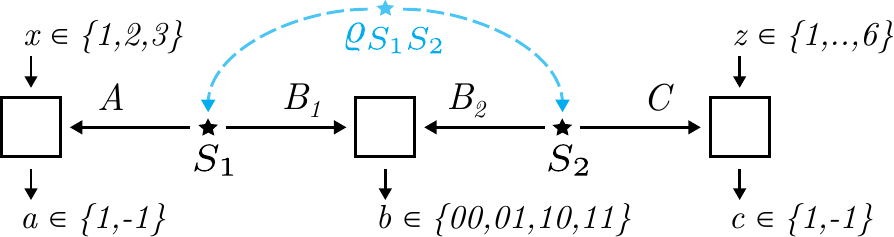}
   \caption{{\bf The setup.}
   There are three parties Alice, Bob and Charlie. Alice and Charlie choose inputs $x, z$ and obtain outcomes $a,c$ respectively. Bob only obtains an outcome $b$. In the original bilocalilty scenario (black), according to {(Real) Quantum Theory}, two independent sources $S_1$ and $S_2$ distribute (real) quantum states $\rho_{AB_1}$ and $\rho_{B_2 C}$ to Alice with Bob and Bob with Charlie, respectively. The parties then perform (real) POVMs $\{ A_{a}^x \}_a$, $\{ B^{b} \}_b$ and $\{ C_z^{c} \}_c$ on their respective subsystems, leading to the distribution $P(a,b,c|x,z)=  \tr\big( (A_x^{a} \otimes B^{b} \otimes C_z^{c})( \rho_{AB_1} \otimes \rho_{B_2 C} )\big)$. In addition, we here allow the sources to initially share some (real) quantum correlations represented by the (real) quantum state $\varrho_{S_1S_2}$ (cyan), such that the state $\rho_{AB_1B_2 C}$ distributed by the source must no longer be product, nor separable, with respect to the bipartition $AB_1|B_2C$. 
   }
   \label{fig:bilocal_setup}
   \end{figure}
In this setup, by measuring their subsystems the three parties Alice, Bob and Charlie generate measurement outcomes following a distribution $P(a,b,c|x,z)$. Now let us consider the linear functional (Bell test\footnote{Notice that  $\mathcal{B}_{00}(P)$ is the  Bell expression from~\cite{pauliselftest} (differently normalized), that is used there to self-test the Pauli measurements. })
\begin{align}
\label{violation}
\mathcal{B}(P)&=\sum_{b\in\{0,1\}^2}\mathcal{B}_b(P)
\qquad \text{with} \\
	\mathcal{B}_b(P)&= (-1)^{b_2}(\cS_{11}^b+\cS^b_{12})+(-1)^{b_1}(\cS^b_{21}-\cS^b_{22})\nonumber\\
	&+(-1)^{b_2}(\cS^b_{13}+\cS^b_{14})-(-1)^{b_1+b_2}(\cS^b_{33}-\cS^b_{34})\nonumber\\
	&+ (-1)^{b_1}(\cS^b_{25}+\cS^b_{26})-(-1)^{b_1+b_2}(\cS^b_{35}-\cS^b_{36}),
	\label{adjusted_3CHSH}
\end{align}
$\cS^b_{xz}=\sum_{a,c=-1,1}P(a,b,c|x,z)ac$ and $b=b_1b_2$. If Alice, Bob and Charlie share an arbitrary tripartite state, the maximal quantum (real or complex) value of $\mathcal{B}(P)$ is given by $\mathcal{B}_{\C}^{\rm sup} := 6 \sqrt{2}$. This can be seen from the fact that Bob's measurement is fixed and can be performed at the source, then,  for each outcome $b$, $\mathcal{B}_b(P)$ is essentially a sum of three CHSH tests between Alice and Charlie, whose individual maximal quantum score is $2\sqrt{2}$.

It was shown in~\cite{renou2021quantum} that when the sources $S_1$ and $S_2$ are independent, the same maximal quantum value $\mathcal{B}_{\C}^{\rm sup}\approx 8.49$ is still achievable in complex quantum theory, while, the optimal real value is bounded by 
\begin{equation}\label{eq: real bound}
     \mathcal{B}_{\R}^{\rm sup} \leq \mathcal{B}_{\R}^{\rm ub} := 7.66 < \mathcal{B}_{\C}^{\rm sup}.
\end{equation} 
Hence the real quantum bound can be violated experimentally, showing that no {model in Real Quantum Theory} with independent sources can explain the observations.  
Importantly, ref.~\cite{renou2021quantum} goes two steps further in the analysis of the thought experiment. 

First, it was shown that in Real Quantum Theory, the maximal score of the test $\mathcal{B}(P)=6 \sqrt{2}$ self-tests the following state shared by Alice and Charlie (after Bob's system is discarded) 
$\bar{\varrho}_{AC} = \frac{1}{2}\left( \Phi^-  + \Psi^+\right)$ with $\ket{\Phi^-} = \frac{1}{\sqrt 2} (\ket{00}-\ket{11}) $ and $\ket{\Psi^+} = \frac{1}{\sqrt 2} (\ket{01}+\ket{10})$, which is entangled in Real Quantum Theory.

Second, it was shown that the real quantum bound~\eqref{eq: real bound} remains valid if the sources are given access to unbounded shared randomness. Hence Real Quantum Theory cannot be rescued by allowing the sources to be classically correlated, i.e. by allowing them to initially share any auxiliary \emph{separable} state $\varrho_{S_1S_2}$ in Fig.~\ref{fig:bilocal_setup}. 

 In the following sections, we ask what can still be concluded if the assumption of source independence is relaxed in a general way, by also allowing the systems sent out by the sources to be partially entangled.

\section{Relaxing the independence assumption}

Consider two quantum systems $S_1$ and $S_2$ prepared in some state $\varrho_{S_1S_2}$, and let ${\tt All}$ denote the set containing all their possible states (in Complex or Real Quantum Theory). The systems are independent if their state is product $\varrho_{S_1S_2} = \varrho_{S_1} \otimes \varrho_{S_2},$ with  $\tt Ind$ denoting the associated set. In quantum theory, independence can be broken in (at least) two ways -- with classical or quantum correlations. We say that two systems are classically correlated if their state is separable  $\varrho_{S_1 S_2} = \sum_i \lambda_i \varrho_{S_1}^i \otimes \varrho_{S_2}^i,$
for $\lambda_i \, \geq 0$ and $\sum_i \lambda_i=1$, and we denote the set of separable states by $\tt Sep$, with $\tt Ind \subsetneq Sep \subsetneq All$. In the latter case of quantum correlations, $\varrho_{S_1 S_2}$ is entangled ($\tt All \setminus Sep$). 

 Recall that in the thought experiment of~\cite{renou2021quantum} independence is not required, instead arbitrary classical correlations of the sources are allowed. Note however, that both in Complex and Real Quantum Theory entangled states can be arbitrarily close to $\tt Ind$, hence proving that an experiment does not admit an interpretation with separable states does not put a lower bound on the (quantum) correlations of the sources. For the following discussion it is thus useful to quantify the amount of correlations and entanglement shared by two systems. In both cases, this can be done by the following geometric measure (see Appendix~\ref{app: ind}) given by the distance to the corresponding set 
\begin{align}
D_{\tt Sep(Ind)} (\varrho_{S_1S_2} ) &:= \inf_{\sigma \in \tt Sep(Ind)} \frac{1}{2}\|\rho-\sigma\|,
\end{align}
with $\|A\|= \tr \sqrt{A^\dag A}$ and the immediate property $D_{\tt Sep} (\varrho_{S_1S_2} ) \leq D_{\tt Ind} (\varrho_{S_1S_2}  )  \leq 1$.  Another natural measure of  entanglement that we consider, is the so-called entanglement of formation $E_{F}(\varrho_{S_1S_2})$~\cite{bennett1996mixed, caves2001entanglement}, giving the minimal entanglement cost of preparing the state. It is worth emphasizing that all the definitions above are meaningful {both in Complex and Real Quantum Theory.}

\subsection{Necessary and sufficient source correlations in Real Quantum Theory (ideal case)}
Recall that in the ideal realization with $\mathcal{B}(P)=6 \sqrt{2}$, the final state of Alice and Charlie in {Real Quantum Theory} (after discarding Bob) can be self-tested to  $\bar{\varrho}_{AC} = \frac{1}{2}\left( \Phi^-  + \Psi^+\right)$. Operationally this means that this state can be extracted from the state $\tr_{B_1B_2} \rho_{AB_1B_2C}$ by means of real local operations (RLO) performed by Alice and Charlie \cite{renou2021quantum}. 
In turn, for our setup in Fig.~\ref{fig:bilocal_setup} it follows that the same state can be locally extracted from the quantum correlations $\varrho_{S_1S_2}$ initially shared by the sources 
\begin{equation}\label{eq:rhos1s2}
   \mathcal{B}(P)=6 \sqrt{2}:\,\, \varrho_{S_1 S_2}  
   \xrightarrow[\text{RLO}]{} 
    \bar{\varrho}_{S_1 S_2} = \frac{1}{2}(\Phi^-  + \Psi^+).
\end{equation}
It is worth noting that the self-testing can be performed sequentially and without assuming that the same state is measured in each experimental round~\cite{Bancal2021selftestingfinite}. In this case one bounds the average fidelity of the extracted states with $\bar{\varrho}_{S_1 S_2}$. For the ideal realization this fidelity can be arbitrary close to one in the limit of large statistics, essentially guaranteeing that a copy of $\bar{\varrho}_{S_1 S_2}$ could have been locally extracted in each round.\footnote{Note that while~\cite{Bancal2021selftestingfinite} focuses on the CHSH test, the general analysis of sequential self-testing presented there is not specific to a particular Bell test.}

In Appendix~\ref{app: ind} we show that in {Real Quantum Theory} this state satisfies $D_{\tt Ind}(\bar{\varrho}_{S_1 S_2})= D_{\tt Sep}(\bar{\varrho}_{S_1 S_2})=\frac{1}{2}$ and $E_F(\bar{\varrho}_{S_1 S_2} )=1$, where both entanglement measures attain their maximal values (on two-\reqbit \ states). Recall that these are measures, non-increasing under real local operations, hence the values also lower bound entanglement and correlations of any state $\varrho_{S_1S_2}$ initially shared by the sources. In fact, as we discuss later, using  the simulation strategy of Ref.~\cite{mckague2009simulating} one can show that these correlations are also sufficient to simulate any quantum experiment in the network of Fig.~\ref{fig:bilocal_setup} {in Real Quantum Theory}. 

\subsection{Necessary source correlations in Real Quantum Theory (general case)}
In an actual experiment we don't expect to observe the maximal value $\mathcal{B}_{\C} = 6  \sqrt{2}$. Thus we study how the entanglement of the state shared by the sources relates to the maximal value $\mathcal{B}(P)$ achievable {in Real Quantum Theory}. 

Let there be a real quantum state $\tau_{AB_1B_2C}$ such that $D_{\tt Sep} (\tau_{AC} ) \leq \epsilon$ for $\tau_{AC} =\tr_{B_1B_2} (\tau_{AB_1B_2C})$, then the value $\mathcal{B}(P)$ of the Bell test in Eq.~\eqref{violation} for this state is bounded by the following hierarchy of semidefinite programs (SDPs):
\begin{align}
	\max_{ \substack{\{ \Gamma^n(\tau^b)\}_b, \Gamma^n(\sigma),\\ \Gamma^n(M), \Gamma^n(N),  {P}}} \, & \mathcal{B}(P)
    \label{eq:moment_problem}
    \\
	\text{s.t.} \, &
    {\small \begin{pmatrix}
		\Gamma^n(M) & \Gamma^n(\tau) -\Gamma^n(\sigma)\\
		\Gamma^n(\tau)^T -\Gamma^n(\sigma)^T & \Gamma^n(N)
	\end{pmatrix}} \geq 0 \nonumber \\
    &\Gamma^n(\sigma) \geq 0, \quad  \Gamma^n(\tau^b) \geq 0 \quad \forall b \ \nonumber\\
	& \Gamma^n(M)_{(1,1)} +\Gamma^n(N)_{(1,1)} \leq  4 \epsilon 
    \nonumber\\
	&\Gamma^n(\tau)_{(1,1)}=\Gamma^n(\sigma)_{(1,1)} = 1 \nonumber\\
	&\Gamma^n(\sigma)^{T_A} = \Gamma^n(\sigma) \nonumber \\
    &{\rm \bf NS}(P) \nonumber
\end{align}
where $\Gamma^n(\tau):=\sum_b\Gamma^n(\tau^b)$, $\tau:= \sum_b \tau^b$  and $\tau^b=\tr_B((\id_A \otimes B^b \otimes \id_C) \tau_{AB_1B_2C})$, $n$ is the level of the hierarchy where $\Gamma^n$ denotes moment matrices up to degree $n$ in the measurement operators of Alice and Charlie. The subscript $_{(i,j)}$ denotes the element at position $(i,j)$ in a matrix and the transposition and partial transposition of system $A$ are denoted $^{T}$ and $^{T_A}$ respectively. Notice that the elements of $P$ are elements of the moment matrices $\Gamma^n(\tau^b)$, so that the objective function is indeed linear in these.

This hierarchy is an extension of the one developed in~\cite{renou2021quantum}, where here we are able to incorporate trace distance constraints into the moment problem. For a derivation 
we refer to Appendix~\ref{app:hierarchy}. 
At level 2, the above hierarchy leads to the upper-bounds on the Bell score displayed in Table~\ref{tab:ep}.
\begin{table}
\centering   
 	\begin{tabular}{|c | c c c c c c c c|} 
 		\hline
        $\mathcal{B}(P)$ & 7.66 & 7.72& 7.78 & 7.88 & 8.06 & 8.22 & 8.37 & 8.50\\
        \hline
        $\epsilon_1\,  [\%]$ & 0& 2.5 & 5 & 10 & 20 & 30 & 40& 50\\
        \hline
         $\epsilon_2\,  [\%]$ & 0& 0.4 & 0.7 & 1.3 & 2.4 & 3.3  & 4.2 & 4.9 \\
 		\hline
 	\end{tabular}
     \caption{{\bf Upper bounds on Bell value for all states with $D_{\tt Sep} (\tau_{AC} ) \leq \epsilon$ in Real Quantum Theory.} Here, $\epsilon_1$ corresponds to the bounds from the moment problem \eqref{eq:moment_problem}, $\epsilon_2$  to the simpler linear bound \eqref{eq:bound2}. This table illustrates that the bounds based on the moment problem considerably outperform the simpler bounds. Notice also that for $\epsilon_1=0.5$ -- the distance to separable of the optimal complex strategy -- we obtain a value of $8.4953$ which coincides with $6\sqrt{2}$ (up to numerical precision), as expected.}
     \label{tab:ep}
\end{table}
These results have been obtained adapting the codes from~\cite{renou2021quantum} to incorporate the additional constraints, the codes are available at~\cite{github_url}.  They were run in MATLAB using the optimisation package YALMIP~\cite{yalmip} and the SDP solver Mosek~\cite{mosek}. Our implementation uses up to about 500 GB of RAM  and thus had to be run on a high performance computing cluster (with a runtime between 5 and 6 days per value). We remark that other solvers like SCS~\cite{scs} are able to solve the SDP on a desktop computer (in around 1-2 hours) but with less accuracy.

In Table~\ref{tab:ep} we compare our results with the linear bound
\begin{equation}\label{eq:bound2}
    D_{\tt Sep}(\tau_{A C}) \geq \frac{\cB(P) - \cB_{\tt Set}^{\rm sup}}{ \cB_{\tt All}^{\rm sup} - \cB_{\tt All}^{\rm inf}},
\end{equation}
derived in Appendix~\ref{app:linear_bounds} by bounding the norm of the Bell operator associated to $\mathcal{B}(P)$, where $\cB_\mathds{\tt Set}^{\rm sup(inf)}$ denotes the maximal (minimal) Bell score achieved by quantum states from $\tt Set$. In our case, for {Real Quantum Theory} we have $\cB_{\tt Sep}^{\rm sup}\leq \mathcal{B}_{\R}^{\rm ub}$ and $\cB_{\tt All}^{\rm sup} =-\cB_{\tt All}^{\rm inf} =6 \sqrt 2$, implying  $\epsilon \geq \frac{\cB( {P} )-7.66}{12 \sqrt 2}.$

Observing a Bell score of $ 7.66 < \mathcal{B}(P)^* < 6  \sqrt{2}$ we can use the above results to conclude how far from separable a real quantum state $\tau_{A C}$ needs to be. Coming back to the setup of Fig.~\ref{fig:bilocal_setup}, it follows that the same bound $D_{\tt Sep}(\varrho_{S_1S_2})\geq \epsilon$ must also hold for the state shared by the sources, since $D_{\tt Sep}$ is non-increasing under local operations.

\section{Discussion}
With our results we are thus able to relax the source independence assumption as well as to show tightness of our bounds in the ideal case. Here, we elaborate more on the resources needed to overcome the discrepancy between Real and Complex Quantum Theory and comment on the role of the assumption itself.

\subsection{Sufficient source correlations in Real Quantum Theory}
\label{sec:suff}
We have shown that if one assumes that the initial entanglement shared between the sources in the setup of Fig.~\ref{fig:bilocal_setup} is bounded, the observed correlations can rule out any description of the experiment {in Real Quantum Theory}. In turn, it is natural to ask what is the minimal entanglement resource sufficient in Real Quantum Theory to simulate the experiment. 

\begin{proposition*} \label{prop}
Any quantum experiment involving $m$ independent sources $S_i$ can be simulated with a real quantum model, if the sources initially share one copy of the $m$ rebit state\footnote{Notice that this state is real despite it being conveniently represented using the complex unit.}
\begin{equation}\label{eq: state REf}
    \bar \varrho_{\bm S}^{(m)} = \frac{1}{2}\left(\bigotimes_{i=1}^m \ketbra{y_+}_{S_i} +\bigotimes_{i=1}^m \ketbra{y_-}_{S_i}\right).
\end{equation}
where  $\bm S = S_1 \dots S_m$ and  $\ket{y_\pm} = (\ket{0}\pm \ii\ket{1})/\sqrt 2$.
\end{proposition*}

The proof of this proposition extends the construction from~\cite{mckague2009simulating} to networks, using that real quantum states can be simultaneously entangled and locally broadcastable (see End Matter).

As notation suggests, in the case of two sources the state $\bar \varrho_{S_1S_2}^{(2)}$ in Eq.~\eqref{eq: state REf} is indeed equal to the state self-tested by the maximal Bell score $\mathcal{B}(P)$ in Eq~\eqref{eq:rhos1s2}. We can thus conclude that in order to reproduce the ideal implementation of the thought experiment~\cite{renou2021quantum}, the sources must pre-share entangled resources which are sufficient to reproduce any experiment in the same network.\\

\subsection{Real quantum physics with entangled space}
Following the real simulation strategy of Ref.~\cite{mckague2009simulating}, we consider each subsystem $\bar{R}_i$ of a global system $\bar{R}$  as composed of subsystems $\bar{R}_i=R_i L_i$, with a rebit $L_i$. The global system given by the collection of rebits ${\bm L} = L_{1}\dots L_n$ plays the role of a {\it real-quantum reference frame of complexness} (see End Matter for more details).  
To derive our bounds, we have so far understood the systems $L_i$ as describing certain local, initially uncorrelated, degrees of freedom. 
This is in line with the natural interpretation that views the systems $L_i$, alongside $R_i$, as describing  certain physical degrees of freedom, whose preparation is controlled by the sources present in the setup, adopted in the present work as well as in~\cite{renou2021quantum} .

One can however be more radical, and delegate this global delocalized reference frame system $\bm L$ to the structure of empty space itself, reminiscent of the entangled vacuum in quantum field theory~\cite{summers1985vacuum,reznik2003entanglement}. In this view every point of space $x$ has access to a local qubit $L_x$. These qubits are prepared in a global entangled state one time for all time (e.g.\ at the big bang), they can be locally broadcasted, and never need to be measured out. Such a `decorated' real quantum theory allows one to represent any state, transformation and measurement of Complex Quantum Theory, and can therefore not be ruled out from experimental observations. Nevertheless, here an ontological argument can be made in favor of Complex Quantum Theory -- it seems excessive to equip space with such an entangled global reference frame structure, 
when it can be avoided with a simple tweak of the number field.

\section{Conclusion}
Our results show that in the setup of Fig.~\ref{fig:bilocal_setup} a global state $ \bar \varrho_{S_1 S_2}$ with one rebit per source is necessary and sufficient for reproducing the predicted complex quantum correlations of~\cite{renou2021quantum} in Real Quantum Theory. Building on  \cite{mckague2009simulating}, we demonstrate that the sufficiency extends to any complex quantum correlations arising in experiments with $m$ independent sources. Finally, focusing on the tripartite Bell test introduced in~\cite{renou2021quantum}, we provide a tradeoff between source independence and Bell score in Real Quantum Theory.

In the past two years, several experiments were performed~\cite{exprealqm1,exprealqm2, exprealqm3} to rule out an explanation according to Real Quantum Theory in the specific scenario considered in~\cite{renou2021quantum}.\footnote{None of them achieves this in a loophole free manner.} The experimental implementations~\cite{exprealqm1} and~\cite{exprealqm3} report Bell values of $8.09(1)$ and $7.83(3)$, respectively. Our results give these findings a quantitative meaning, in the sense that they imply that to implement this in Real Quantum Theory one would need $D_{\tt Sep}({\varrho}_{S_1 S_2})$ above $0.2$ and $0.05$, respectively.\footnote{The experiment~\cite{exprealqm2} used a different Bell expression tailored to their setup, so that an assessment of the required $D_{\tt Sep}({\varrho}_{S_1 S_2})$ will require a separate analysis with the tools introduced in this work.} In fact, performing experiments with independent sources is challenging, especially for setups implemented on a single chip like in Ref.~\cite{exprealqm1}. Thus, analyzing the effects of partial source independence is an interesting research direction, independently of the specific question considered in this paper.

In summary, we demonstrated that Real Quantum Theory can only reproduce the quantum correlations observed in these experiments if some amount of entanglement, bounded from the correlations, is available to the sources before the experiment. 
Since assuming the presence of some hidden background entanglement seems ontologically excessive, concluding the inadequacy of Real Quantum Theory is predicated on guaranteeing sufficient partial independence of the sources.

\acknowledgements
 This work was supported by the Swiss National Science Foundation  via Ambizione PZ00P2\_208779, NCCR-SwissMap and by the Swiss Secretariat for Education,
Research and Innovation (SERI) under contract number UeM019-3.

\section*{End Matter}

{\it{Proof of Proposition}--} To see this we start by briefly recalling the real simulation strategy of Ref.~\cite{mckague2009simulating}. Here one sets $d'_i=2d_i$ (in $\ket{\Psi}_{\bar{\bm  R} }$above), such that $\R^{2 d_i }\simeq \R^{d_i}\otimes \R^{2}$ and each system can be decomposed into two subsystems $\bar{R}_i=R_i L_i$, with a rebit $L_i$. The global system given by the collection of rebit ${\bm L} = L_{1}\dots L_n$ plays the role of a {\it real-quantum reference frame of complexness}, it is encoded in the logical rebit subspace spanned by the entangled states $\ket{\mathcal R}_{\bm L}\propto \bigotimes_i \ket{y_+}_{L_i}+\bigotimes_i \ket{y_-}_{L_i}$ and  $\ket{\mathcal I}_{\bm L}\propto \ii \bigotimes_i \ket{y_+}_{L_i}- \ii \bigotimes_x \ket{y_-}_{L_i}$.  The crucial property of this encoding, is that the action of the real unitary operator $
\mathds{J}:=\ii\,  \sigma_Y=\ketbra{0}{1}-\ketbra{1}{0}$ on any of the rebits $L_i$ mimics the action of the imaginary unit on complex Hilbert spaces, i.e. $\mathds{J}_{L_i} \ket{\mathcal{R}}_{\bm L} = \ket{\mathcal{I}}_{\bm L}$ and $ \mathds{J}_{L_i} \ket{\mathcal{I}}_{\bm L} = -\ket{\mathcal{R}}_{\bm L}$. {Hence, in the logical subspace of the reference frame, identified by the projector $\id^{\Pi}_{\bm L}=  \ketbra{\mathcal{R}}{\mathcal{R}}_{\bm L}+\ketbra{\mathcal{I}}{\mathcal{I}}_{\bm L}$, we find $\id^{\Pi}_{\bm L} \mathds{J}_{L_i} \id^{\Pi}_{\bm L} =: \mathds{J}_{\bm L} = \ketbra{\mathcal{I}}{\mathcal{R}}_{\bm L}-\ketbra{\mathcal{R}}{\mathcal{I}}_{\bm L}$.}

The addition of such a delocalized reference frame system allows one to simulate any complex quantum model. In particular, any complex state $\ket{\Psi}_{\bm C}$ (see above) 
can be represented by $\ket{\Psi}_{\bm R\bm L} = \ket{\Psi^{\rm Re}}_{\bm R}\ket{\mathcal R}_{\bm L} + \ket{\Psi^{\rm Im}}_{\bm R}\ket{\mathcal I}_{\bm L}$. As already noted in \cite{mckague2009simulating} this state  is however not invariant under the global phase transformation of the complex state $e^{\ii \varphi} \ket{\Psi}_{\bm C}$, and, in fact, also not under discarding of unused reference frame qubits. An invariant real simulation model for states can be obtained by using the `dephased'  real density operator 
\begin{align}
    \varrho_{\bm R \bm L}\!&:=\!\frac{1}{2}(\Psi_{\bm R \bm L} + \mathds{J}_{\bm L} \Psi_{\bm R \bm L} \mathds{J}_{\bm L}^T) =\frac{1}{2}\left( \varrho^{\rm Re}_{\bm R} \otimes\id^{\Pi}_{\bm L}
    +\varrho^{\rm Im}_{\bm R} \otimes \mathds{J}_{\bm L}\right) 
    \label{eq:ref_state}
\end{align} 
representing the projector $\Psi_{\bm C} = \ketbra{\Psi}_{\bm C}$.
It is direct to verify that in this case, the marginal state of the reference frame is always given by
\begin{equation}
    \bar \varrho_{\bm L}^{(n)} =\tr_{\bm R}\,  \varrho_{\bm R\bm L} \, = \frac{1}{2} \id^{\Pi}_{\bm L},
    \label{eq:ref_state2}
\end{equation}
which is identical to Eq.~\eqref{eq: state REf} upon replacing $n$ with $m$ and $\bm L$ with $\bm S$.  For completeness, we discuss this simulation in full detail in Appendix~\ref{app: simulation}. There we also give the simulation strategy for an arbitrary measurement, and show that it can be done by only measuring the systems belonging to $\bm R$, i.e. without the need to ever consume and re-entangle the reference frame $\bm L$.

Next, let us consider a quantum network, where $m$  sources distribute quantum systems to $n$ parties $A_i$. The simulation model guarantees that providing the sources and parties with an entangled reference frame in the initial state $\bar{\varrho}^{(n+m)}_{\bm S \bm L}$ yields a universal real quantum simulation model. To proceed further with the proof, we demonstrate in Appendix~\ref{app: simulation} that in {Real Quantum Theory} any state  $\bar \varrho_{\bm L}^{(n)}$ exhibits two seemingly contradictory properties. It contains one ebit of (bound) entanglement  of formation across any bipartition, while at the same time being \emph{locally broadcastable}\footnote{The two properties are indeed contradictory in complex quantum theory according to the no-local-broadcasting theorem \cite{Luoon,Pianino}, displaying another operational difference between the two theories.}. In particular, with local real operations the sources can transform $\bar \varrho_{\bm S}^{(m)}$ into $\bar \varrho_{\bm S \bm L}^{(m+n)}$, with each system $L_i$ prepared by one of the sources connected to it. Using Eq.~\eqref{eq: state REf} one can easily verify that this global state indeed satisfies $\varrho_{\bm S }^{(n)}=\tr_{\bm L} \bar \varrho_{\bm S \bm L}^{(n+m)}$ and $\varrho_{\bm L }^{(m)}=\tr_{\bm S} \bar \varrho_{\bm S \bm L}^{(n+m)}$. Hence, given the initial state $\bar \varrho_{\bm S}^{(m)}$, the sources can prepare and distribute the reference frame $\bm L$ (alongside other degrees of freedom) to all parties involved in the experiment. After broadcasing, the marginal state of the sources $\bar \varrho_{\bm S}^{(m)}$ remains unchanged, and can be reused in subsequent rounds of the experiment (without correlating the rounds).

Finally, note that this construction is not specific to networks as in Fig.~\ref{fig:bilocal_setup}, but can be repeated for any setup involving $m$ independent sources. 

\bigskip

{{\it{Note added}--} Following the publication of this letter appeared two papers with a seemingly contradicting conclusion -- that quantum theory with real Hilbert spaces is on par with the usual quantum theory~\cite{hita2025quantum,hoffreumon2025quantum}. In fact, the constructions presented in both papers are formally similar to the general real simulation strategy underlying our Proposition. The key difference is the interpretation of this construction. While we see this  either as coming from a relaxation of the independence assumption in Real Quantum Theory, or as coming from a 'decorated' version of Real Quantum Theory with a shared reference frame, the authors of~\cite{hita2025quantum,hoffreumon2025quantum} give it different interpretations. In \cite{hita2025quantum}, the authors modify the usual definition of states (for the reference frame system, i.e. flag qubits in their terminology), such that the state $\bar \varrho_{\bm L}^{(n)}$ appears product and does not violate their notion of source independence. 
In \cite{hoffreumon2025quantum}, the authors directly impose the state structure of the form given in the rhs of Eq.~\eqref{eq:ref_state} $\Gamma[\rho_{\bm C}]=\left(\begin{array}{cc}
\rho_{\bm R}^{\rm Re} & -\rho_{\bm R}^{\rm Im}\\
\rho_{\bm R}^{\rm Im} & \rho_{\bm R}^{\rm Re}
\end{array}\right)$, and modify the standard tensor product composition rule towards one tailored to mimic Complex Quantum theory $\Gamma[\rho_{\bm C}]\odot\Gamma[\sigma_{\bm C'}]=\id \otimes (\rho_{\bm R}^{\rm Re}\otimes \sigma_{\bm R'}^{\rm Re}-\rho_{\bm R}^{\rm Im}\otimes \sigma_{\bm R'}^{\rm Im}) +\mathds{J} \otimes (\rho_{\bm R}^{\rm Re}\otimes \sigma_{\bm R'}^{\rm Im}+\rho_{\bm R}^{\rm Im}\otimes \sigma_{\bm R'}^{\rm Re})$. 
    
    In both cases the authors construct new versions of quantum theory over real Hilbert spaces that embed Complex Quantum Theory. In this sense, those works differ from ours and from~\cite{renou2021quantum} in that they focus on reformulations of Complex Quantum Theory, while we are focusing on the (experimental) distinction between Real and Complex Quantum Theories.}

\bibliographystyle{unsrt}
\bibliography{RQM.bib}

\begin{appendix}

\onecolumngrid
    
\section{Properties of independence and entanglement measures}

\label{app: ind}

\subsection{Monotonicity and convexity of trace distance in Real Quantum Theory}

It is well known that the trace distance $ D(\rho,\sigma) := \frac{1}{2}\|\rho -\sigma\|$, with trace norm $\|A\|= \tr \sqrt{A^\dag A}$,  
is monotonic, i.e. satisfies the data processing inequality
\begin{align}
D(\rho,\sigma) &\geq  D\big(\cE(\rho),\cE(\sigma)\big) 
\end{align}
for any CPTP map $\cE$.  The same is true within real quantum physics, as the proof goes through identically with real Hilbert spaces. In particular, a  textbook way to show monotonicity is to argue that the trace distance equals the `classical' total variation distance between probability distributions induced by the optimal measurement, hence it cannot be increased by $\cE$ which can always be seen as part of the measurement (its adjoint is unital, see discussion of CP maps over real Hilbert spaces in the next section). This can be shown in two steps. First the upper bound is established using the Cauchy-Schwartz inequality. Second, the measurement saturating the bound is found by diagonalizing $\rho -\sigma$. None of these steps is specific to complex Hilbert spaces. Note that here and below CP maps $\cE$ do not necessarily preserve the dimension of the quantum system on which they act, and may, for example, include tracing out  subsystems as a specific case.

In turn the joint convexity of the trace distance $ D(\sum_i  p_i \rho^i,\sum_i p_i \, \sigma^i) \leq  \sum p_i \,D(\rho^i, \sigma^i)$ directly follows from the triangle inequality for the trace norm (both over real and complex Hilbert spaces). 

\subsection{Local operations (with classical communication) and entanglement in Real Quantum Theory}

To have a proper discussion of correlations and entanglement as resources~\cite{chitambar2019quantum}, in addition to the corresponding sets of free states $\tt Ind$ and $\tt Sep$ in Real and Complex Quantum Theories
 one needs to introduce the associated sets of free operations. In the case of entanglement this is given by local operations and classical communication (LOCC)~\cite{chitambar2014everything}, while for  correlations it is natural to identify the set of free operations as all local operations (LO). These classes of transformation are well understood in the context of usual quantum theory. Nevertheless, one may wonder how much of our understanding pertains when considering Real Quantum Theory. In this respect, the following proposition is  reassuring.

\begin{proposition*}
\cite{blecher2021real}\cite{chiribella2023positive}
For any completely positive (CP) map $\cE_{\R}: L(\R^{d_i})\to L(\R^{d_o})$  its complexification $\cE_\C: L(\C^{d_i})\to L(\C^{d_o})$, defined by $\widetilde \cE_\C[X +\ii Y] :=  \cE_\R[X]  +\ii \, \cE_\R[Y]$ for any $X,Y\in L(\R^{d_i})$, is also CP.
\end{proposition*}
\begin{proof}
    See corollary 4.2 in \cite{chiribella2023positive}, or lemma 2.3 in \cite{blecher2021real}.
\end{proof}
 The proposition guarantees that there are no CP maps $\cE_{\R}$ over real Hilbert space, which cease to be CP when viewed as maps over complex Hilbert spaces (i.e. when complexified)\footnote{As a side remark, note that the same statement is not true for positive maps, with counterexamples constructed in \cite{chiribella2023positive}.}. This is important because CP maps correspond to physically allowed transformations, in particular the most general LO is a quantum instrument represented by a collection of CP trace-non-increasing maps. Note also, that the somehow converse statement is trivial. More precisely, let us call a map $\cE_\C$ real if it satisfies $\cE_\C[X] = (\cE_\C[X])^*$ for all $X\in L(\R^{d_i})$. Then, any real CP map $\cE_\C$ induces a map $\cE_\R$ (with $\widetilde \cE_\C =  \cE_\C $) which is also CP (conditions for $\cE_\R$ being CP are also necessary for $\cE_\C$).  In summary, we can now identify the set of CP maps $\cE_\R$ over real Hilbert spaces with the subset of real CP maps $\cE_\C$ over corresponding complex Hilbert spaces. The same identification naturally translates to local operations (quantum instruments), but also to LOCC by considering sequences of LO interjected with rounds of classical communication.

Next, let us quickly verify that our free operations preserve the free sets of states both for Real and Complex Quantum Theory. LO manifestly preserve independence $(\cE_A\otimes \cE_B)[\varrho_A\otimes \varrho_B] = \cE_A[\varrho_A]\otimes \cE_B[\varrho_B]$, while a round of CC may break independence but can only establish classical correlations
$$ {\rm CC }: \quad \left(\sum_i p_i\, \varrho_A^i \otimes \ketbra{i}_{R_A}\right) \otimes \varrho_B \mapsto \sum_i p_i  \left (\varrho_A^i \otimes \ketbra{i}_{R_A}\right) \otimes \left (\varrho_B \otimes \ketbra{i}_{R_B} \right).$$

Finally, note that in Real Quantum Theory all  separable states have the property 
\begin{equation}
 \varrho_{AB} \in {\tt Sep} \implies \varrho_{AB}^{T_A} = \sum_i \lambda_i (\varrho_A^i)^T \otimes \varrho_B^i = \sum_i \lambda_i\,  \varrho_A^i \otimes \varrho_B^i = \varrho_{AB},
\end{equation}
that we use later.

\subsection{General properties of geometric measures}

Let $D(\rho,\sigma)$ be a monotonic divergence on quantum states, and $\tt Free\subset {\tt All}$ be the set of (free) states. Consider a CPTP map $\cE$ which describes a free transformation, i.e. such that
\begin{equation}
 \sigma \in {\tt Free} \implies \cE[\sigma] \in {\tt Free}
\end{equation}
where it is implicit that the density operators $\sigma$ must be chosen in the domain of $\cE$. For a state $\rho$ define the distance (more generarally divergence) to the free set as 
\begin{equation}
    D_{\tt Free}(\rho) := \inf_{\sigma \in \tt Free} D(\rho,\sigma),
\end{equation}
where it is implicit that $\sigma$ must be chosen in the domain of $D(\rho,\bullet)$, i.e. $\rho$ and $\sigma$ must describe states of the same quantum systems. By definition $D_{\tt Free}(\sigma)=0$ on free states. Let us also show that this quantity is non-increasing under free operations $\cE$. We have
\begin{align}
     D_{\tt Free}(\cE(\rho)) &= \inf_{\sigma' \in \tt Free} D(\cE(\rho),\sigma')
       \leq\inf_{\sigma \in \tt Free} D(\cE(\rho),\cE(\sigma))
      \leq\inf_{\sigma \in \tt Free} D(\rho,\sigma)
      = D_{\tt Free}(\rho).
\end{align}
where we use the fact that $\cE(\sigma)$ is free in the second inequality, and the monotonicity of $D$ in the third.\\

Now, let $D(\rho,\sigma)$ be the trace distance. It follows that $D_{\tt Sep}(\rho)$ ($D_{\tt Ind}(\rho)$) satisfy the two essential properties~\cite{chitambar2019quantum} for a measure of entanglement (correlations): they are zero for separable (product) states and non-increasing under LOCC (LO).

\subsection{Entanglement of formation}

For a pure bipartite state $\Psi_{AB}$ the entropy of entanglement is given by the  von Neumann entropy of the marginal states
$$
E_{H}(\Psi_{AB}) := H(\tr_A \Psi_{AB}) = H(\tr_B \Psi_{AB}).
$$
This definition can not be simply lifted to mixed states, as in this case marginal entropy does not need to come from entanglement. The standard way to proceed is thus to consider its convex roof, called entanglement of formation~\cite{bennett1996mixed},
\begin{align}
E_{F}(\varrho_{AB} ) := \inf_{\{p_k \Psi^{(k)}_{AB}\}} &\sum_k p_k\,E_{H}\left(\Psi_{AB}^{(k)}\right)  \\
\text{s.t.} &\sum_k p_k \, \Psi^{(k)}_{AB} = \varrho_{AB},
\end{align}
where $p_k >0 \ \forall k$ and all $\Psi^{(k)}_{AB}$ are pure states.

It is well known that $E_F(\rho_{AB})$ is an entanglement measure in complex quantum theory~\cite{bennett1996mixed}. Unsurprisingly this remains true in Real Quantum Theory, as we now argue. The fact that it is still zero on separable states is immediate. To see that it is non-increasing under LOCC one can follow the steps of the original proof~\cite{bennett1996mixed}, established in the context of complex quantum theory. The proof proceeds in two steps: first, showing that for any bipartite pure state, the entanglement of formation can not increase by LO applied by one party (follows from convexity of the von Neumann entropy); second, showing that this is still true for mixed states (essentially follows from convexity of $E_F$). Both remain true in Real Quantum Theory.

\subsection{Maximal distances to separable two qubit states}

Let us now find the maximal distance that a two (real) qubit states can have to the separable set. Let $\varrho_{AB}=\sum_i p_i \, \Psi_{AB}^i$ be such a state and its spectral decomposition, we have
\begin{align}
     D_{\tt Sep}(\varrho_{AB}) &= \inf_{\sigma_{AB}\in \tt Sep} D(\varrho_{AB} ,\sigma_{AB})  = \inf_{\sigma_{AB}^i 
    \in {\tt Sep}} D\left( \varrho_{AB}, \sum_i p_i \sigma_{AB}^i \right )\\
    &= \inf_{\sigma_{AB}^i 
    \in {\tt Sep}} D\left( \sum_i p_i \, \Psi_{AB}^i , \sum_i p_i \sigma_{AB}^i \right )\\
    & \leq \inf_{\sigma_{AB}^i 
    \in {\tt Sep}} \left( \sum_i p_i \, D(\Psi_{AB}^i , \sigma_{AB}^i ) \right)
\end{align}
where we used joint convexity. Now since the sum has at most 4 entries and all $\sigma_{AB}^i$ can be varied independently, we can permute the sum and the inf leading to 
\begin{align}
     D_{\tt Sep}(\varrho_{AB}) 
    & \leq  \sum_i p_i \, \inf_{\sigma_{AB}^i 
    \in {\tt Sep}}  D(\Psi_{AB}^i , \sigma_{AB}^i ).
\end{align}
It remains to upper bound the distance of an arbitrary pure state to the separable set, i.e. $ \inf_{ 
    \sigma_{AB}\in \tt Sep} D\left( \Psi_{AB},
    \sigma_{AB} \right)$.

To do so consider the Schmidt decomposition of the pure state $\Psi_{AB}$, in Complex and Real Quantum Theory, it guarantees that for some choice of local bases the state can be written as 
\begin{equation}
    \ket{\Psi}_{AB} = c \ket{00}_{AB} + s \ket{11}_{AB} 
\end{equation}
with real non-negative $ c, s$ satisfying $c^2+s^2=1$. For the separable state $ \sigma_{AB} = c^2 \ketbra{00} + s^2 \ketbra{11}$ we then find 
\begin{equation}
    D(\Psi_{AB}, \sigma_{AB}) = cs \frac{1}{2}\Big\| ( \ketbra{00}{11} +\ketbra{11}{00})\Big\|= c s \leq \frac{1}{2},
\end{equation}
and hence $D_{\tt Sep}(\varrho_{AB})\leq \frac{1}{2}$. The bound is saturated for Bell states, as follows from the next section.

\subsection{Distance to separable of $\bar{\rho}_{S_1 S_2}$}
For the state $\bar{\varrho}_{S_1 S_2}$ from (5) in the main text, the distance $D_{\tt Sep}(\bar{\varrho}_{S_1 S_2})$ to the real separable states can be lower bounded by the following SDP: 
\begin{align}
	\min_{ \sigma} \, & \frac{1}{4} \left(\tr(M)+\tr(N)\right)
    \\
	\text{s.t.} \, &
    {\begin{pmatrix}
	M & \bar{\varrho}_{S_1 S_2} -\sigma\\
	(\bar{\varrho}_{S_1 S_2}-\sigma)^T & N
\end{pmatrix}} \geq 0 \nonumber \\
    &\sigma \geq 0, \quad \tr(\sigma) =1  \nonumber\\
    &\sigma^{T_{S_1}}=\sigma.
\end{align}
The first two lines encode the trace distance (see e.g.~\cite{Watrous_Book}) where, since $\bar{\varrho}_{S_1 S_2},\sigma$ are real, we use the transpose $^T$ and $M$, $N$ are symmetric matrices. The third line encodes that $\sigma$ is a (real) quantum state and the fourth encodes that it is real separable (where $^{T_{S_1}}$ denotes partial transposition of system $S_1$)~\cite{caves2001entanglement}.\footnote{Notice that this condition is only known to be tight for two rebits~\cite{caves2001entanglement}, in general it is not.} 

Solving this SDP we obtain 
$$D_{\tt Sep}(\bar{\varrho}_{S_1 S_2}) = \frac{1}{2}.$$
Notice further that for the maximally mixed state, which is product, $\tilde{\sigma}=\frac{\id}{4}$ the distance is $\frac{1}{2} \|{\bar{\varrho}_{S_1 S_2} - \tilde{\sigma}}\| = \frac{1}{2}$, hence 
$$D_{\tt Sep}(\bar{\varrho}_{S_1 S_2})=D_{\tt Ind}(\bar{\varrho}_{S_1 S_2})=\frac{1}{2}.$$

Notice that the four Bell states have the same distance to separable, which can be verified by running the same SDP.

\subsection{Real entanglement of formation of the sate $\bar{\varrho}_{S_1 S_2}$}

\label{sec: EF of rho}

An closed-form expression of the real entanglement of formation of two \reqbit states $\varrho_{AC}$ has been derived in \cite{caves2001entanglement}
\begin{equation}
    E_F(\varrho_{AC}) = H\left( \frac{1+\sqrt{1-|\tr \varrho_{AC} (\sigma_Y\otimes \sigma_Y)_{AC}|^2}}{2} \right).
\end{equation}
A straightforward computation gives for our state of interest $\bar{\varrho}_{S_1 S_2} = \frac{1}{2}(\Phi^-  + \Psi^+)_{AC}$ 
    \begin{align}
        E_F(\bar{\varrho}_{S_1 S_2}) =1. 
    \end{align}
    Clearly this is also the maximal possible value, since  the entropy of entanglement is bounded by one for two \reqbit states.

\section{Semidefinite programming hierarchy for $\epsilon$ independence}
\label{app:hierarchy}

We aim to maximise the Bell value $\mathcal{B}(P)$ over all $P$ satisfying
$ P(a,b,c|x,z)=  \tr((A_x^a \otimes B^b \otimes C_z^c) \tau_{AB_1B_2C})$  for some real POVMs $A_x, B, C_z$ and a real state $\tau_{AB_1B_2C}$ satisfying $D_{\tt Sep}(\tau_{AB_1B_2C}) \leq \epsilon$. Note that this implies $D_{\tt Sep}(\tau_{AC}) \leq \epsilon$ for $\tau_{AC}=\sum_b \tau^b$, where $\tau^b=\tr_B((\id_A \otimes B^b \otimes \id_C) \tau_{AB_1B_2C})$. 

To arrive at a relaxation of this problem we follow~\cite{Moroder, renou2021quantum}. Let us consider the operator lists corresponding to Alice and Bob's measurement outcomes, 
\begin{align*}
	\mathcal{A} &= \left\{ A_1:= A_{1|1}, A_2:= A_{1|2}, A_3:=A_{1|3} \right\} \\
	\mathcal{C} &= \left\{C_{1|1}, C_{1|2}, C_{1|3}, C_{1|4}, C_{1|5}, C_{1|6} \right\}.
\end{align*}
Now let us take the list of monomials up to degree $n$ of elements in $\mathcal{A}$ including the identity to define the list of monomials $\mathcal{A}^n$. Let us label each monomial by the tuple $\alpha$ of measurement labels of its operators, e.g.  for $A_{1}A_3 A_1$, we obtain the monomial $M_\alpha$ with label $\alpha=(1,3,1)$, the identity has label $\alpha=(0)$.
Now let us define an orthonormal set of  vectors $\{\ket{\alpha}\}_\alpha$ 
and choose a basis $\{ \ket{k} \}_k$ for the Hilbert space that the operators in $\mathcal{A}$ act on. Then we can define the map $$\Omega_A^n(\rho)=\sum_k K_k^n \rho \ (K_k^n)^\dagger, $$
with Kraus operators $K_k^n=\sum_{\alpha} \ket{\alpha}  \bra{k} M_\alpha$,
where $M_\alpha$ is the monomial with label $\alpha$.
For $C$ we define an analogous map $\Omega_C^n$. Due to the Kraus representation, the map $\Omega_A^n \otimes \Omega_C^n$ is completely positive, and thus for any quantum state $\rho_{AC}$ leads to a positive matrix
$$\Gamma^n(\rho_{AC})=\Omega_A^n \otimes \Omega_C^n(\rho_{AC}),$$
where $\Gamma^n(\rho_{AC})_{(1,1)}= \tr(\rho_{AC})=1$. Notice further that some elements of $\Gamma^n(\rho_{AC})$ are related, namely whenever $M_{\alpha_1} M_{\alpha_2}^\dagger=M_{\alpha_3} M_{\alpha_4}^\dagger$ and $M_{\gamma_1} M_{\gamma_2}^\dagger=M_{\gamma_3} M_{\gamma_4}^\dagger$, then the respective elements of the moment matrix are equal, i.e., $\bra{\alpha_1} \bra{\gamma_1} \Gamma^n(\rho_{AC}) \ket{\alpha_2} \ket{\gamma_2}=\bra{\alpha_3} \bra{\gamma_3} \Gamma^n(\rho_{AC}) \ket{\alpha_4} \ket{\gamma_4}$ (which has to be explicitly imposed in (6) in the main text). 

Notice furthermore that $\Gamma^n(\tau)= \sum_b \Gamma^n(\tau^b)$, where the $\Gamma^n(\tau^b)$ are positive, $\Gamma^n(\tau)_{(1,1)}=1$, and that all elements $P(1,b,1|x,z)$ and the marginals $P(b)$, $P(1,b|x)$ $P(b,1|z)$ appear as elements in the $\Gamma^n(\tau^b)$, thus setting those elements to fixed values.

In addition, $P$ is a non-signalling distribution, thus satisfying the constraints $\sum_a P(a,b,c|x,z)=P(b,c|z)$, $\sum_c P(a,b,c|x,z)=P(a,b|x)$, $\sum_a P(a,b|x)=\sum_c P(b,c|z)=P(b)$, $\sum_b P(b) =1$ and $P(a,b,c|x,z) \geq 0$. We summarize all these constraints as ${\rm \bf NS}(P)$.

We furthermore note that for a real separable state $\sigma_{AC}=\sum_\lambda p(\lambda) \sigma_a^\lambda \otimes \sigma_c^\lambda$ we have that
$$\Gamma^n(\sigma_{AC}) 
= \sum_\lambda p(\lambda) \Omega_A^n(\sigma_A^\lambda) \otimes \Omega_C^n(\sigma_C^\lambda), $$
is also real separable (as the completely positive maps $\Omega_A^n,\Omega_C^n$ are both real). This implies that $\Gamma^n(\sigma_{AC})^{T_A}=\Gamma^n(\sigma_{AC})$~\cite{caves2001entanglement,renou2021quantum}, which can also be checked by direct calculation.

In addition, notice that $\| \tau -\sigma \| \leq \epsilon$ implies that \begin{align}
\begin{pmatrix}
	M & \tau -\sigma\\
	(\tau-\sigma)^\dagger & N
\end{pmatrix}
\geq 0 \nonumber \\
\tr{M}+\tr{N} \leq 2 \epsilon, \label{eq:trace_distance}
 \end{align}
where  
$M,N$ are Hermitian matrices~\cite{Watrous_Book}. For real states $\tau, \sigma$, the conjugate transpose reduces to the transpose and the matrices $M,N$ can further be chosen real, since for any matrix $X$, $X \geq 0$ implies that also $X^T \geq 0$. Applying this to the matrix in \eqref{eq:trace_distance} implies that we can replace $M$ with $(M+M^T)/2$, which is a real matrix. The analogous argument holds for $N$.

Now let us consider $$ \id_2 \otimes \Omega_A^n \otimes \Omega_C^n \begin{pmatrix}
M & \tau -\sigma\\
(\tau-\sigma)^T & N
\end{pmatrix} = \begin{pmatrix}
	\Gamma^n(M) & \Gamma^n(\tau -\sigma)\\
	\Gamma^n((\tau-\sigma))^T & \Gamma^n(N)
\end{pmatrix}, $$
which is positive as the map $\id_2 \otimes \Omega_A^n \otimes \Omega_B^n$ is completely positive. Furthermore, $\Gamma^n(M)_{(1,1)}=\tr(M)$ and $\Gamma^n(N)_{(1,1)}=\tr(N)$.

Overall, this leads us to the semidefinite programming problem (6) of the main text.

\section{Linear bound on trace distance from Bell score via operator norm}
\label{app:linear_bounds}

Consider a $n$-partite scenario where $n$ players perform measurements with input $x_i$ to produce an output $a_i$, giving rise to the behaviors
\begin{equation}
\text{P}(a_1,\dots, a_n|x_1\dots x_n).    
\end{equation}
We denote the input and output strings  with $\bm a =(a_1,\dots, a_n)$ and $\bm x= (x_1,\dots, x_n)$.  Let $\mathcal{B}$ be a linear Bell test which associates a score to any behavior $\cB(\text{P})\in \mathds{R}$.\\

Next, let $\tt Set$ denote some quantum state of the systems measured in the Bell test (there can be more that one system per party in a network scenario). We denote with $\cB_\mathds{\tt Set}^{\rm sup}$ and $\cB_\mathds{\tt Set}^{\rm inf}$ the maximal and minimal scores of the Bell test that can be obtained 
\begin{align}
\cB_\mathds{\tt Set}^{\rm sup} &= \sup_{\sigma \in \tt Set, {\rm POVMs}} \cB({\rm P})\\    
\cB_\mathds{\tt Set}^{\rm inf} &= \inf_{\sigma \in \tt Set, {\rm POVMs}} \cB({\rm P}),   
\end{align}
where $\text{P}$ is obtained by applying arbitrary measurements (consistent with a black-box description of the setup) on some state $\sigma$ from $\tt Set$.\\

\begin{proposition*} Consider a linear Bell test $\cB$ and a behavior ${\rm P}(\bm a |\bm x )$ leading to a Bell score $\cB({\rm P})> \cB_\mathds{\tt Set}^{\rm sup}$, for some quantum state set $\tt Sep$. 
Then, the quantum state $\rho$ underlying the behaviour must satisfy
\begin{equation}
    D_{\tt Set}(\rho) \geq \frac{\cB({\rm P}) - \cB_\mathds{\tt Set}^{\rm sup}}{ \cB_\mathds{\tt All}^{\rm sup} - \cB_\mathds{\tt All}^{\rm inf}},
\end{equation}
where $\tt All$ is the set of all density operators.\\
\end{proposition*}

\begin{proof} The behavior ${\rm P}(\bm a |\bm x )$ is obtained by applying some measruemtns on the state $\rho$. Fixing these measuremnts in the Bell test defines a Bell operator $\rm B$
such that 
$$
\tr\, {\rm B} \, \rho = \cB({\rm P}).
$$
Performing the same measurements on any state $\sigma \in \tt Set$ we are guaranteed to get
$\tr\, {\rm B} \, \sigma \leq\cB_\mathds{\tt Set}^{\rm sup}$, 
leading to the bound
\begin{equation}
\tr \, {\rm B}(\rho - \sigma) \geq \cB({\rm P}) - \cB_\mathds{\tt Set}^{\rm sup}.
\end{equation}
Here, the difference of the two density operators can always be decomposed as
$\rho - \sigma = P - Q$, where $P$ and $Q$ satisfy $0\preceq P,Q \preceq 1$ with $\tr P = \tr Q = \frac{1}{2}\| \rho -\sigma\| = D(\rho,\sigma)$, and are thus proportional to density operators $\varrho_P =\frac{P}{ D(\rho,\sigma)}$ and $\varrho_Q= \frac{Q}{ D(\rho,\sigma)}$. We thus find
\begin{align}
    \tr \, {\rm B}\, (\rho - \sigma) &= \tr\, {\rm B}\, (P-Q) = D(\rho,\sigma)\,  \tr\ {\rm B}\,(\varrho_P-\varrho_Q) \leq D(\rho,\sigma)  ( \cB_\mathds{\tt All}^{\rm sup} - \cB_\mathds{\tt All }^{\rm inf}).
\end{align} 
Combining the above bounds we find that
$$
D(\rho,\sigma) \geq \frac{\cB({\rm P}) - \cB_\mathds{\tt Set}^{\rm sup}}{ \cB_\mathds{\tt All}^{\rm sup} - \cB_\mathds{\tt All}^{\rm inf}}
$$
must hold for all states $\sigma$ in $\tt Set$, proving the result.
\end{proof}

Again it is worth noting that the statement is valid both in Complex and Real Quantum Theories.

\section{The real simulation model of~\cite{mckague2009simulating}}

\label{app: simulation}
Here we discuss the simulation of complex quantum models with reals ones, following \cite{mckague2009simulating}, and display some additional properties.

Consider $n$ quantum systems $\bm C = C_1 \dots C_n$ described by $d_i$ dimensional quantum systems. A pure state of the systems is then represented by a vector in the associated Hilbert space
\begin{equation}\label{eq: complex state}
    \ket{\Psi}_{\bm C} \in \C^{d_1}\otimes \dots \otimes\C^{d_n}.
\end{equation}
Mixed states are represented by density operator $\rho_{\bm C}$, evolution and measurement by completely positive (CP) maps $\cE_{\bm C'}$ and POVMS $\{E_{\bm C'}^x\}$ acting on subsystems $\bm C' $ of $\bm C$. The goal is to present a quantum model with the same number of real quantum systems $\bar {\bm R} = \bar R_1 \dots \bar R_n$ represented on real Hilbert spaces $\R^{d_1' }\otimes \dots \otimes\R^{d_n' }$, capable to simulate the physics of the complex quantum models with the same `computational graph'. That is, operations acting on $\bm C'$ must act on the same subsystem $\bar {\bm R}'$ in both models.

To do so, following~\cite{mckague2009simulating} we set $d_i' =2 d_i$ to obtain $\R^{2 d_i }\simeq \R^{d_i}\otimes \R^{2}$ and decompose each local system according to this partition as $\bar{R}_i=R_i L_i$. The global system given by the collection of rebits ${\bm L} = L_{1}\dots L_n$ plays the role of a real-quantum reference frame of complexness, and is encoded in a logical rebit subspace defined by the two entangled states $\{\ket{\mathcal R}_{\bm L}, \ket{\mathcal I}_{\bm L}\}$. The states can be conveniently expressed as 
\begin{align}
     \ket{\mathcal R}_{\bm L} &:= \frac{1}{\sqrt 2}\big(\bigotimes_i \ket{y_+}_{L_i}+\bigotimes_i \ket{y_-}_{L_i}\big) \\
     \ket{\mathcal I}_{\bm L} &:= \frac{\ii}{\sqrt 2}\big(\bigotimes_i \ket{y_+}_{L_i}-\bigotimes_x \ket{y_-}_{L_i}\big),
\end{align}
where $\ket{y_\pm} = (\ket{0}\pm \ii\ket{1})/\sqrt 2$ are the (complex) eigenstates of the Pauli $\sigma_Y$ operator, while both global states above are real. This encoding has the property that the action of the real unitary operator $
\mathds{J}:=\ii \, \sigma_Y=\ketbra{0}{1}-\ketbra{1}{0}$ on any of the rebits $L_j$ mimics the action of the imaginary unit on complex Hilbert spaces. More precisely, one can verify that 
\begin{equation}\label{eq: IR sim}
\mathds{J}_{L_j} \ket{\mathcal{R}}_{\bm L} = \ket{\mathcal{I}}_{\bm L} \quad \text{and} \quad \mathds{J}_{L_j} \ket{\mathcal{I}}_{\bm L} = -\ket{\mathcal{R}}_{\bm L}.
\end{equation}
The addition of such a delocalized reference frame system allows one to simulate complex quantum models with real ones using the following representation of states and operators
\begin{align}
  \label{eq: state sim pure}  \ket{\Psi}_{\bm C} &\simeq \ket{\Psi}_{\bm R \bm L} := \ket{\Psi^{\rm Re}}_{\bm R}\ket{\mathcal R}_{\bm L} + \ket{\Psi^{\rm Im}}_{\bm R}\ket{\mathcal I}_{\bm L} \\
  \label{eq: op} 
    M_{\bm C'}& \simeq M_{\bm R'\bm L'}:= M_{\bm R'}^{\rm Re} \otimes \id_{\bm L'} + M_{\bm R'}^{\rm Im} \otimes \mathds{J}_{\bm L'}
\end{align}
with  $\bra{\Psi}_{\bm C} \simeq \bra{\Psi}_{\bm R \bm L}$ and $ M_{\bm C'}^\dag \simeq  M_{\bm R'\bm L'}^T  = (M_{\bm R'}^{\rm Re})^T \otimes \id_{\bm L'} - (M_{\bm R'}^{\rm Im})^T \otimes \mathds{J}_{\bm L'}$ . Here, the superscripts denote the real and imaginary part in the computational  basis and $ \mathds{J}_{{\bm L}'}$ can be realized by any $ \mathds{J}_{ L_i}$ for any choice of qubit $L_i$ in $\bm L'$. For completeness let us explicitly verify the consistency of the representation (which follows from Eq.~\eqref{eq: IR sim})
\begin{align}
\ket{\Omega}_{\bm C} &= M_{\bm C'} \ket{\Psi}_{\bm C} \\  \simeq \ket{\Omega}_{\bm R \bm L} &= \ket{\Omega^{\rm Re}}_{\bm R}\ket{\mathcal R}_{\bm L} + \ket{\Omega^{\rm Im}}_{\bm R}\ket{\mathcal I}_{\bm L}\\
& = (M_{{\bm R}'}^{\rm Re}\ket{\Psi^{\rm Re}}_{\bm R} -M_{{\bm R}'}^{\rm Im}\ket{\Psi^{\rm Im}}_{\bm R} )\ket{\mathcal R}_{\bm L} + (M_{{\bm R}'}^{\rm Re}\ket{\Psi^{\rm Im}}_{\bm R} + M_{{\bm R}'}^{\rm Im}\ket{\Psi^{\rm Re}}_{\bm R} )\ket{\mathcal I}_{\bm L} \\
&= M_{\bm R'\bm L'} \ket{\Psi}_{\bm R\bm L}
\\
\braket{\Omega}_{\bm C} &= \braket{\Omega^{\rm Re}}_{\bm R} + \braket{\Omega^{\rm Im}}_{\bm R} =\braket{\Omega}_{\bm R \bm L}.
\end{align}
 One also verifies that in the simulation unitaries (Hermitian) operators become orthogonal (symmetric), while sets of operators describing POVMs and Kraus operators (representing a CP map) remain POVMs and Kraus operators respectively. Hence the representation in Eqs.~(\ref{eq: state sim pure},\ref{eq: op}) allows one to simulate any complex quantum model with a real quantum one, provided that all parties, or nodes in the computational graph, have access to a qubit from the reference frame.  

\subsection{Global-phase-invariant state representation}

 Finally, following \cite{mckague2009simulating}, one notices that the real state representation in Eq.~\eqref{eq: state sim pure} is not invariant under global phase transformations of the complex state $\ket{\Psi}_{\bm C}\to e^{\ii \varphi} \ket{\Psi}_{\bm C}$. In particular, we find that 
\begin{align}
  \ii \ket{\Psi}_{\bm C} &\simeq \mathds{J}_{\bm L} \ket{\Psi}_{\bm R \bm L} := \ket{\Psi^{\rm Re}}_{\bm R}\ket{\mathcal I}_{\bm L} - \ket{\Psi^{\rm Im}}_{\bm R}\ket{\mathcal R}_{\bm L}.
\end{align}
A simple way to obtain an invariant real representation is to explicitly dephase the real state 
\begin{equation}
\ketbra{\Psi}_{\bm R \bm L} \mapsto \mathcal{D}\left [\ketbra{\Psi}_{\bm R \bm L}\right]= \frac{1}{2}(\ketbra{\Psi}_{\bm R \bm L} + \mathds{J}_{\bm L} \ketbra{\Psi}_{\bm R \bm L} \mathds{J}_{\bm L}^T),
\end{equation}
or equivalently preparing the reference frame with an auxiliary qubit and tracing it out $\tr_{L'} \ketbra{\Psi}_{\bm R \bm LL'} =  \mathcal{D}\left [\ketbra{\Psi}_{\bm R \bm L}\right]$. 
After including such dephasing the representation of complex density operators becomes
\begin{align}\label{eq:inv rep}  
\varrho_{\bm C} \simeq  \varrho_{\bm R \bm L} := \varrho^{\rm Re}_{\bm R} \otimes \left(\frac{ \ketbra{\mathcal{R}}+\ketbra{\mathcal{I}}}{2}\right)_{\bm L} +\varrho^{\rm Im}_{\bm R} \otimes \left(\frac{ \ketbra{\mathcal{I}}{\mathcal{R}}-\ketbra{\mathcal{R}}{\mathcal{I}}}{2}\right)_{\bm L}. 
\end{align}
By construction the state keeps its form after tracing out unused reference frame rebits $ \tr_{L'} \varrho_{\bm R \bm L L' } =  \varrho_{\bm R \bm L}$.

Two additional remarks are worth adding. First, remark that the marginal state of the reference frame remains unchanged (since  $\tr \varrho_{\bm R}^{\rm Re}=1$ and $\tr \varrho_{\bm R}^{\rm Im}=0$)
\begin{align}\label{eq: rho gen}
    \bar \varrho_{\bm L}^{(n)} &:= \tr_{\bm R}\,   \varrho_{\bm R \bm L} =\left(\frac{ \ketbra{\mathcal{R}}+\ketbra{\mathcal{I}}}{2}\right)_{\bm L} 
    = \frac{1}{2}\left(\bigotimes_i \ketbra{y_+}_{L_i} +\bigotimes_i \ketbra{y_-}_{L_i}\right).
\end{align}

Second, the dephasing map $\mathcal{D}$ commutes with all the valid transformation in the real simulation model, i.e. those representing a complex equivalent, hence it can be applied to the initial state of the reference frame. In other words, to run the simulation strategy it is sufficient for the parties to share the reference frame in the state  $\bar \varrho_{\bm L}^{(n)} $ in Eq.~\eqref{eq: rho gen}.

\subsection{The reference frame state is bound entangled}

Consider a bipartition of all the frame qubits $\bm L = \bar{\bm L} \bar{\bm L}^C$. With respect to it, the reference frame state $\bar \varrho_{\bm L}^{(n)} $ in Eq.~\eqref{eq: rho gen} reads
\begin{align}\label{eq: ref state bi}
    \bar \varrho^{(n)}_{ \bar{\bm L} \bar{\bm L}^C} &= \frac{1}{2} \left(\ketbra{\bm y_+}_{\bar{\bm L}} \otimes \ketbra{\bm y_+}_{\bar{\bm L}^C} + \ketbra{\bm y_-}_{\bar{\bm L}} \otimes \ketbra{\bm y_-}_{\bar{\bm L}^C} \right)
\end{align}
where $\ket{\bm y_\pm}_{\bar{\bm L}} = \bigotimes_{L_i\in \bar{\bm L}} \ket{y_\pm}_{L_i}$. Expressed in the product basis $\{\ket{\mathcal{R},\mathcal{R}}_{\bar{\bm L} \bar{\bm L}^C} ,\ket{\mathcal{R},\mathcal{I}}_{\bar{\bm L} \bar{\bm L}^C} ,\ket{\mathcal{I},\mathcal{R}}_{\bar{\bm L} \bar{\bm L}^C} ,\ket{\mathcal{I},\mathcal{I}}_{\bar{\bm L} \bar{\bm L}^C} \}$ the state reads 
 \begin{equation}
    \bar \varrho^{(n)}_{\bar{\bm L} \bar{\bm L}^C} = \left(
\begin{array}{cccc}
 \frac{1}{4} & 0 & 0 & -\frac{1}{4} \\
 0 & \frac{1}{4} & \frac{1}{4} & 0 \\
 0 & \frac{1}{4} & \frac{1}{4} & 0 \\
 -\frac{1}{4} & 0 & 0 & \frac{1}{4} \\
\end{array}
\right) = \frac{1}{2}(\Phi^-  + \Psi^+) .
 \end{equation}
 We have seen in Appendix~\ref{sec: EF of rho} that this state has (across the bipartition) one ebit of entanglement of formation ($E_F(\bar \varrho^{(n)}_{\bar{\bm L} \bar{\bm L}^C})=1$) in Real Quantum Theory. 

In addition, note that this entanglment is bound. To see this note that in complex quantum theory the state is separable, hence no entanglement can be distilled by LOCC from any number of copies of the state. This conclusion remains true in Real Quantum Theory, since real-LOCC transformations remain LOCC when complexified.

\subsection{The reference frame state is locally broadcastable}

Again, consider a bipartition of all the frame qubits $\bm L = \bar{\bm L}\bar{\bm L}^C$, with $m<n$ qubits in the subsystem $\bar{\bm L}$. From Eq.~\eqref{eq: ref state bi} it is straightforward to verify that 
\begin{equation}
    \tr_{\bar{\bm L}^C}\, \bar \varrho^{(n)}_{\bar{\bm L}\bar{\bm L}^C} =\bar \varrho^{(m)}_{\bar{\bm L}}.
\end{equation}\\

Now, starting with $\bm L$ in the state $\bar \varrho^{(n)}_{\bm L}$ bring a fresh qubit $L_1'$ in the state $\ket{0}_{L_1' }$ in contact with $L_1$, and perform the following joint orthogonal transformation
\begin{align}\label{eq: isom}
    U_{L_1 L_1'} &= \ketbra{y_+}_{L_1} \otimes \left( \ketbra{y_+}{0}+\ketbra{y_-}{1}\right)_{L_1' } + \ketbra{y_-}_{L_1} \otimes \left( \ketbra{y_-}{0}+\ketbra{y_+}{1}\right)_{L_1' }\\
    &= \left(
\begin{array}{cccc}
 \frac{1}{\sqrt{2}} & \frac{1}{\sqrt{2}} & 0 & 0 \\
 0 & 0 & \frac{1}{\sqrt{2}} & -\frac{1}{\sqrt{2}} \\
 0 & 0 & \frac{1}{\sqrt{2}} & \frac{1}{\sqrt{2}} \\
 -\frac{1}{\sqrt{2}} & \frac{1}{\sqrt{2}} & 0 & 0 \\
\end{array}
\right)
\end{align}
where the matrix is written in the product computational basis. By Eq.~\eqref{eq: isom} we have 
\begin{equation}\label{eq: local iso}
U_{L_1 L_1'} \ket{y_+}_{L_1}\ket{0}_{L_1' }= \ket{y_+}_{L_1}\ket{y_+}_{L_1' } \quad \text{and}  \quad U_{L_1 L_1'} \ket{y_-}_{L_1}\ket{0}_{L_1' } = \ket{y_-}_{L_1}\ket{y_-}_{L_1' }
\end{equation}
implying
\begin{equation}
    U_{L_1 L_1'} \left(\bar \varrho^{(n)}_{\bm L}\otimes \ketbra{0}_{L_1' }\right) U_{L_1 L_1'}^T = \varrho^{(n+1)}_{\bm L L_1'}
\end{equation}
One can thus locally expand the reference frame state to include one more \reqbit. This procedure can be repeated by all the parties leading to 
\begin{equation}
    \bigotimes_i U_{L_i L_i'} \left(\bar \varrho^{(n)}_{\bm L}\otimes \ketbra{\bm 0}_{\bm L'  }\right)    \bigotimes_i U_{L_i L_i'}^T = \bar \varrho^{(2n)}_{\bm L \bm L'},
\end{equation}
and we already know that $\tr_{\bm L' } \varrho^{(2n)}_{\bm L \bm L'} = \varrho^{(n)}_{\bm L }$ and $\tr_{\bm L } \varrho^{(2n)}_{\bm L \bm L'} = \varrho^{(n)}_{\bm L'}$, demonstrating local broadcasting~\cite{Pianino} of the entangled state $\bar \varrho^{(n)}_{\bm L}$.
Note that in complex quantum theory local broadcasting is only possible for the so-called classical-classical correlated states~\cite{Pianino} -- a subset of separable states describing correlated classical registers. This is consistent with  $\bar \varrho_{\bm L}^{(n)} 
    = \frac{1}{2}\left(\bigotimes_i \ketbra{y_+}_{L_i} +\bigotimes_i \ketbra{y_-}_{L_i}\right)$ describing $n$ copies of a random bit upon complexification. The possibility of locally broadcasting entangled states in Real Quantum Theory exhibits yet another operational difference with the complex quantum theory.

Finally, let us also verify that the local broadcasting of the reference frame works for the global state $\varrho_{\bm R \bm L}$ in Eq.~\eqref{eq:inv rep}. To shorten the notation we define the real CPTP map (isometry)
\begin{equation}
    {\rm LBC}_{L_i}[\bullet]:= U_{L_i L_i' } \left(\bullet \otimes \ketbra{0}_{L_i' } \right) U_{L_i L_i' }^T,
\end{equation}
and compute
\begin{align}
     {\rm LBC}_{L_i}[\varrho_{\bm R \bm L}] &= \varrho^{\rm Re}_{\bm R} \otimes  {\rm LBC}_{L_i}\left[\left(\frac{ \ketbra{\mathcal{R}}+\ketbra{\mathcal{I}}}{2}\right)_{\bm L} \right]+\varrho^{\rm Im}_{\bm R} \otimes  {\rm LBC}_{L_i}\left[ \left(\frac{ \ketbra{\mathcal{I}}{\mathcal{R}}-\ketbra{\mathcal{R}}{\mathcal{I}}}{2}\right)_{\bm L}\right]\\
     &= \varrho^{\rm Re}_{\bm R} \otimes \left(\frac{ \ketbra{\mathcal{R}}+\ketbra{\mathcal{I}}}{2}\right)_{\bm L L_i'} +\varrho^{\rm Im}_{\bm R} \otimes  \left(\frac{ \ketbra{\mathcal{I}}{\mathcal{R}}-\ketbra{\mathcal{R}}{\mathcal{I}}}{2}\right)_{\bm L}\\
    &= \varrho_{\bm R \bm L L_i' }
\end{align}
using the fact Eqs.~\eqref{eq: local iso} imply   $U_{L_i L_i'}\ket{\mathcal R}_{\bm L}\ket{0}_{L_i' } = \ket{\mathcal R}_{\bm L L_i'}$ and $U_{L_i L_i'}\ket{\mathcal I}_{\bm L}\ket{0}_{L_i' } = \ket{\mathcal I}_{\bm L L_i'}$.

\subsection{Measurement simulation does not consume the reference frame entanglement}

Finally, let us demonstrate that in the real simulation all complex measurements can be simulated without consuming the entanglement present in the reference frame. The most straightforward way to see this is via the local broadcasting of the reference frame state. 

Consider a composite system $\bm C=\bar{\bm C} \bar{\bm C}^C$ prepared in a global state $\varrho_{\bm C} =\varrho_{\bar{\bm C} \bar{\bm C}^C}$, and an arbitrary measurement performed on the subsystem $\bar{\bm C}$ localized in space~\footnote{We are only interested in simulating localizable measurements~\cite{beckman2001causal,eggeling2002semicausal,pauwels2024classification}, which have a physical interpretation consistent with the principle of locality. Those can be decomposed as local measurements plus nonlocal resources -- entanglement and quantum communication (if allowed by the causal structure underlying the setup).}. It is described by a POVM $\{E_{\bar{\bm C}}^{(a)}\}$ on the associated complex Hilbert space. A possible strategy to simulate it in the real model is as follows. Starting form $\varrho_{\bm R \bm L}=\varrho_{\bar{\bm R}\bar{\bm R}^C \bar{\bm L}\bar{\bm L}^C}$ one first locally broadcasts the reference frame qubits $\bar{\bm L}$ to obtain the  state  $\varrho_{\bar{\bm R}\bar{\bm R}^C \bar{\bm L}\bar{\bm L}' \bar{\bm L}^C}$. Then the real POVM 
\begin{equation}
    E_{\bar{\bm R}\bar{\bm L}'}^{(a)} = E_{\bar{\bm R}}^{(a,\rm Re)}\otimes \id_{\bar{\bm L}'} + E_{\bar{\bm R}}^{(a,\rm Im)}\otimes \mathds{J}_{\bar{\bm L}'}
\end{equation}
is performed on the systems $\bar{\bm R}\bar{\bm L}'$. Noting that 
\begin{align}
\tr_{\bar{\bm L}'}\left(\frac{ \ketbra{\mathcal{R}}+\ketbra{\mathcal{I}}}{2}\right)_{\bm L \bar{\bm L}'} &=- \tr_{\bar{\bm L}'} \, \mathds{J}_{\bar{\bm  L}'} \left(\frac{ \ketbra{\mathcal{I}}{\mathcal{R}}-\ketbra{\mathcal{R}}{\mathcal{I}}}{2}\right)_{\bm L \bar{\bm L}'}= \left(\frac{ \ketbra{\mathcal{R}}+\ketbra{\mathcal{I}}}{2}\right)_{\bm L}\\
\tr_{\bar{\bm L}'} \left(\frac{ \ketbra{\mathcal{I}}{\mathcal{R}}-\ketbra{\mathcal{R}}{\mathcal{I}}}{2}\right)_{\bm L \bar{\bm L}'}&=\tr_{\bar{\bm L}'} \, \mathds{J}_{\bar{\bm  L}'}\left(\frac{ \ketbra{\mathcal{R}}+\ketbra{\mathcal{I}}}{2}\right)_{\bm L \bar{\bm L}'} = \left(\frac{ \ketbra{\mathcal{I}}{\mathcal{R}}-\ketbra{\mathcal{R}}{\mathcal{I}}}{2}\right)_{\bm L}
\end{align}
one verifies that the resulting state of the systems $\bar{\bm R}^C\bar{\bm L}\bar{\bm L}^C$ is given by
\begin{align}
    \varrho_{\bar{\bm R}^C \bar{\bm L}\bar{\bm L}^C}^{(a)} &= \tr_{\bar{\bm R}\bar{\bm L}'} \, \left( E_{\bar{\bm R}\bar{\bm L}'}^{(a)} \varrho_{\bar{\bm R}\bar{\bm R}^C \bar{\bm L}\bar{\bm L}^C} \right)\\
    \\
    &= \tr_{\bar{\bm R}} (E_{\bar{\bm R}}^{(a,\rm Re)} \varrho_{\bm R}^{\rm Re} - E_{\bar{\bm R}}^{(a,\rm Im)} \varrho_{\bm R}^{\rm Im}) \otimes \left(\frac{ \ketbra{\mathcal{R}}+\ketbra{\mathcal{I}}}{2}\right)_{\bm L} \\
    &+ \tr_{\bar{\bm R}} (E_{\bar{\bm R}}^{(a,\rm Re)} \varrho_{\bm R}^{\rm Im} + E_{\bar{\bm R}}^{(a,\rm Im)} \varrho_{\bm R}^{\rm Re}) \otimes  \left(\frac{ \ketbra{\mathcal{I}}{\mathcal{R}}-\ketbra{\mathcal{R}}{\mathcal{I}}}{2}\right)_{\bm L}.
\end{align}
The final state $\varrho_{\bar{\bm R}^C \bar{\bm L}\bar{\bm L}^C}^{(a)}$ contains all the reference frame rebits present intitially and indeed simulates $\varrho_{\bar{\bm C}^C}^{(a)} = \tr_{\bar{\bm C}} E_{\bar{\bm C}}^{(a)} \varrho_{\bm C}$. This shows that all complex quantum measurements can be represented in the real simulation model without consuming the reference frame state. 
\end{appendix}

\end{document}